\documentclass[11pt]{article}
\usepackage[cmex10]{amsmath}
\usepackage{amsthm}
\usepackage{amssymb,amsfonts,textcomp}
\usepackage{latexsym}
\usepackage[nospace, compress]{cite}
\usepackage[noend,ruled]{algorithm2e}

\usepackage{graphicx} 
\usepackage{latexsym}
\usepackage{multirow}
\usepackage{rotating}

\newcommand{\fixlist}{\addtolength{\itemsep}{-5pt}}

\newcommand{\capacity}{{{{\mathrm{cap}}}}}
\newcommand{\cost}{{{{\mathrm{cost}}}}}
\newcommand{\pos}{{{{\mathrm{pos}}}}}
\newcommand{\bordainc}{{{{\mathrm{B, inc}}}}}
\newcommand{\bordadec}{{{{\mathrm{B, dec}}}}}
\newcommand{\OPT}{{{{\mathrm{OPT}}}}}
\newcommand{\opt}{{{{\mathrm{opt}}}}}
\newcommand{\sat}{{{{\mathrm{sat}}}}}
\newcommand{\E}{\mathop{\mathbb E}}
\newcommand{\w}{{{{\mathrm{w}}}}}

\title{Fully Proportional Representation as Resource Allocation: Approximability Results}
\author{Piotr Skowron \\ University of Warsaw \\ Warsaw, Poland
\and
        Piotr Faliszewski\\ AGH University \\ Krakow, Poland
\and 
        Arkadii Slinko\\University of Auckland\\ Auckland, New Zealand
}

\newtheorem{theorem}{Theorem}
\newtheorem{definition}{Definition}
\newtheorem{proposition}[theorem]{Proposition}

\newtheorem{lemma}[theorem]{Lemma}

\newcommand{\np}{{\mathrm{NP}}}
\newcommand{\fpt}{{\mathrm{FPT}}}
\newcommand{\wone}{{\mathrm{W[1]}}}
\newcommand{\wtwo}{{\mathrm{W[2]}}}
\newcommand{\p}{{\mathrm{P}}}

\newcommand{\naturals}{{{\mathbb{N}}}}
\newcommand{\calL}{{{\mathcal{L}}}}
\newcommand{\calA}{{{\mathcal{A}}}}
\newcommand{\calF}{{{\mathcal{F}}}}
\newcommand{\calS}{{{\mathcal{S}}}}

\begin{document}

\maketitle

\begin{abstract}
  We model Monroe's and Chamberlin and Courant's multiwinner voting
  systems as a certain resource allocation problem. We show that for
  many restricted variants of this problem, under standard
  complexity-theoretic assumptions, there are no constant-factor
  approximation algorithms.  Yet, we also show cases where good
  approximation algorithms exist (briefly put, these variants
  correspond to optimizing total voter satisfaction under Borda
  scores, within Monroe's and Chamberlin and Courant's voting
  systems).
\end{abstract}

\section{Introduction}\label{sec::introduction}
Resource allocation is one of the most important issues in multiagent
systems, equally important both to human societies and to artificial
software agents~\cite{ley-sho:b:multiagent-systems}. For example, if
there is a set of items (or a set of bundles of items) to distribute
among agents then we may use one of many auction mechanisms (see,
e.g.,~\cite{ley-sho:b:multiagent-systems,nis-rou-tar-vaz:b:agt} for an
introduction and a review, and numerous recent papers on auction
theory for current results). However, typically in auctions if an
agent obtains an item (a resource) then this agent has exclusive
access to it. In this paper we consider resource allocation for items
that can be shared, and we are interested in computing (approximately)
optimal assignments (in particular, for settings where resource
allocation boils down to multiwinner voting). As opposed to a large
body of research on auctions, resource allocation, and mechanism
design, we do not make any strategic considerations.

Let us explain our resource allocation problem through an
example. Consider a company that wants to provide free sport classes
to its employees. We have a set $N = \{1, \ldots, n\}$ of employees
and a set $A = \{a_1,\ldots, a_m\}$ of classes that are
offered. Naturally, not every class is equally appealing to each
employee and, thus, each employee orders the classes from the most
desirable one to the least desirable one. For example, the first
employee might have preference order $a_1 \succ a_3 \succ \cdots \succ
a_m$, meaning that for him or her $a_1$ is the most attractive class,
$a_3$ is second, and so on, until $a_m$, which is least
appealing. Further, each class $a_i$ has some maximum capacity
$\capacity_{a_i}$, that is, a maximum number of people that can
comfortably participate, and a cost, denoted $c_{a_i}$, of opening the
class (independent of the number of participants).  The company wants
to assign the employees to the classes so that it does not exceed its
sport-classes budget $B$ and so that the employees' satisfaction is
maximal (or, equivalently, their dissatisfaction is minimal).

There are many ways to measure (dis)satisfaction.  For example, we may
measure an employee's dissatisfaction as the position of the class to
which he or she was assigned in his or her preference order (and
satisfaction as $m$ less the voter's dissatisfaction). Then, we could
demand that, for example, the maximum dissatisfaction of an employee
is as low as possible (minimal satisfaction is as high as possible; in
economics this corresponds to egalitarian social welfare) or that the
sum of dissatisfactions is minimal (the sum of satisfactions is
maximal; this corresponds to the utilitarian approach in economics).

It turns out that our model generalizes two well-known multiwinner
voting rules; namely, those of Monroe~\cite{monroeElection} and of
Chamberlin and Courant~\cite{ccElection}.  Under both these rules
voters from the set $N$ submit preference orders regarding
alternatives from the set $A$, and the goal is to select $K$
candidates (the representatives) best representing the voters. For
simplicity, let us assume that $K$ divides $\|N\|$.\footnote{We stress
  that this assumption does not really affect our results. Our
  algorithms would maintain their quality without the assumption. On
  the other hand, without the assumption modeling Monroe's and
  Chamberlin and Courant's systems would be more tedious and some
  calculations would be a bit more involved.} Under Monroe's rule we
have to match each selected representative to $\frac{\|N\|}{K}$ voters
so that each voter has a unique representative and so that the sum of
voters' dissatisfactions is minimal (dissatisfaction is, again,
measured by the position of the representative in the voter's
preference order). Chamberlin and Courant's rule is similar except
that there are no restrictions on the number of voters a given
alternative represents (in this case it is better to think of the
alternatives as political parties rather than particular
politicians). It is easy to see that both methods are special cases of
our setting: For example, for Monroe it suffices to set the ``cost''
of each alternative to be $1$, to set the budget to be $K$, and to set
the ``capacity'' of each alternative to be $\frac{\|N\|}{K}$.  We can
consider variants of these two systems using different measures of
voter (dis)satisfaction, as indicated above (see also the works of
Potthoff and Brams \cite{potthoff-brams}, Betzler et
al.~\cite{fullyProportionalRepr} and of Lu and
Boutilier~\cite{budgetSocialChoice}).

Unfortunately, it is well-known that both Monroe's method and
Chamberlin and Courant's method are $\np$-hard to compute in
essentially all nontrivial
settings~\cite{complexityProportionalRepr,budgetSocialChoice,fullyProportionalRepr}.
This holds even if various natural parameters of the election are
low~\cite{fullyProportionalRepr}. Notable exceptions include, e.g.,
the case where $K$ is bounded by a fixed constant and the case where
voter preferences are single-peaked~\cite{fullyProportionalRepr}.

Nonetheless, Lu and Boutilier~\cite{budgetSocialChoice}---starting
from a very different motivation and context---propose to rectify the
high computational complexity of Chamberlin and Courant's system by
designing approximation algorithms. In particular, they show that if
one focuses on the sum of voters' satisfactions, then there is a
polynomial-time approximation algorithm with approximation ratio $(1 -
\frac{1}{e}) \approx 0.63$ (i.e., their algorithm outputs an
assignment that achieves no less than about $0.63$ of optimal voter
satisfaction).  Unfortunately, total satisfaction is a tricky
measure. For example, under standard Chamberlin and Courant's system,
a $\frac{1}{2}$-approximation algorithm is allowed to match each voter
to an alternative somewhere in the middle of this voter's preference
order, even if there is a feasible solution that matches each voter to
his or her most preferred candidate. On the other hand, it seems that
a 2-approximation focusing on total dissatisfaction would
give results of very high quality.

The goal of this paper is to provide an analysis of our resource
allocation scenario, focusing on approximation algorithms for the
special cases of Monroe's and Chamberlin and Courant's voting
systems. We obtain the following results:
\begin{enumerate}
\fixlist
\item Monroe's and Chamberlin and Courant's systems are hard to
  approximate up to any constant factor for the case where we measure
  dissatisfaction, irrespective of whether we measure the total
  dissatisfaction (Theorems~\ref{theorem:noApprox1}
  and~\ref{theorem:noApprox3}) or the dissatisfaction of the most
  dissatisfied voter (Theorems~\ref{theorem:noApprox2}
  and~\ref{theorem:noApprox4}).

\item Monroe's and Chamberlin and Courant's systems are hard to
  approximate within any constant factor for the case where we measure
  satisfaction of the least satisfied voter
  (Theorems~\ref{theorem:noApprox5}
  and~\ref{theorem:noApprox6}). However, there are good approximation
  algorithms for total satisfaction---for the Monroe's system we
  achieve approximation ratio arbitrarily close to $0.715$ (and often
  a much better one; see Section~\ref{sec:monroe-approx}). For
  Chamberlin and Courant's system we give a polynomial-time
  approximation scheme (that is, for each $\epsilon$, $0 < \epsilon <
  1$, we have a polynomial-time $(1-\epsilon)$-approximation
  algorithm; see Theorem~\ref{theorem:ptas}).
\end{enumerate}

Our work is similar to several lines of research on computational
social choice and multiagent systems. In particular, there is a
well-established line of work on the hardness and approximability of
winner determination for single-winner voting rules, with results for,
for example, Dodgson's
rule~\cite{bar-tov-tri:j:who-won,hem-hem-rot:j:dodgson,car-cov-fel-hom-kak-kar-pro-ros:c:dodgson,car-kak-kar-pro:c:dodgson-acceptable,fal-hem-hem:j:multimode},
Kemeny's
rule~\cite{bar-tov-tri:j:who-won,hem-spa-vog:j:kemeny,ail-cha-new:j:kemeny-approx,cop-fla-rud:j:kemeny-approx,ken-sch:c:kemeny-few-errors},
Young's
rule~\cite{rot-spa-vog:j:young,car-cov-fel-hom-kak-kar-pro-ros:c:dodgson},
and Ranked Pairs method~\cite{bri-fis:c:ranked-pairs}.  Hardness of
winner determination for multiwinner voting rules was studied by
Procaccia, Rosenschein, and Zohar~\cite{complexityProportionalRepr},
by Lu and Boutilier~\cite{budgetSocialChoice}, and by Betzler, Slinko
and Uhlman~\cite{fullyProportionalRepr}. Lu and Boutilier, while
starting from a different context, initiated the study of
approximation algorithms in this setting, which we continue and extend
in this paper.

In the context of resource allocation, our model closely resembles
multi-unit resource allocation with single-unit demand~\cite[Chapter
11]{ley-sho:b:multiagent-systems} (see also the work of Chevaleyre et
al.~\cite{Chevaleyre06issuesin} for a survey of the most fundamental
issues in the multiagent resource allocation theory). The problem of
multi-unit resource allocation is mostly addressed in the context of
auctions (and so it is referred in the literature as multi-unit
auctions); in contrast, we consider the problem of finding a solution
maximizing the social welfare given the agents' preferences. More
generally, our model, is similar to resource allocation with sharable
indivisible goods~\cite{Chevaleyre06issuesin, AiriauEndrissAAMAS2010}.
The most substantial difference is that we require each agent to be
assigned to exactly one alternative. Also, in the context of resource
allocation with sharable items, it is often assumed that the agents'
satisfaction is affected by the number of agents using the
alternatives (the congestion on the alternatives). This forms a class
of problems that are closely related to congestion
games~\cite{rosenthal73congestion}.  Finally, it is worth mentioning
that in the literature on resource allocation it is common to consider
other criteria of optimality, such as
envy-freeness~\cite{Lipton:2004:AFA:988772.988792}, Pareto optimality,
Nash equilibria~\cite{AiriauEndrissAAMAS2010}, and others.

Finally, we should mention that our paper is very close in spirit
(especially in terms of the motivation of the resource allocation
problem) to the recent work of Darmann et al.~\cite{dar-elk-kur-lan-sch-woe:t:group-activity}.

\section{Preliminaries}\label{sec:prelim}
We first define basic notions such as, e.g., preference orders and
positional scoring rules. Then we present our resource allocation
problem in full generality and discuss restrictions modeling Monroe's
and Chamberlin and Courant's voting systems. Finally, we briefly
recall relevant notions regarding computational complexity
theory. \medskip

\noindent
\textbf{Alternatives, Profiles, Positional Scoring Functions.}\quad
For each $n \in \naturals$, we take $[n]$ to mean $\{1, \ldots, n\}$.
We assume that there is a set $N = [n]$ of \emph{agents} and a set $A
= \{a_{1}, \dots a_{m}\}$ of \emph{alternatives}. Each agent $i$ has
\emph{weight} $w_{i} \in \naturals$, and each alternative $a$ has
\emph{capacity} $\capacity_a \in \naturals$ and \emph{cost} $c_a \in
\naturals$. The weight of an agent corresponds to its size (measured
in some abstract way). An alternative's capacity gives the total
weight of the agents that can be assigned to it, and its cost gives
the price of selecting the alternative (the price is the same
irrespective of the weight of the agents assigned to the alternative).
Further, each agent $i$ has a \emph{preference order} $\succ_i$ over
$A$, i.e., a strict linear order of the form $a_{\pi(1)} \succ_{i}
a_{\pi(2)} \succ_{i} \dots \succ_{i} a_{\pi(m)}$ for some permutation
$\pi$ of $[m]$. For an alternative $a$, by $\pos_i(a)$ we mean the
position of $a$ in $i$'th agent's preference order. For example, if
$a$ is the most preferred alternative for $i$ then $\pos_i(a) = 1$,
and if $a$ is the most despised one then $\pos_i(a) = m$.  A
collection $V = (\succ_1, \ldots, \succ_n)$ of agents' preference
orders is called a \emph{preference profile}.  We write $\calL(A)$ to
denote the set of all possible preference orders over $A$. Thus, for
preference profile $V$ of $n$ agents we have $V \in \calL(A)^n$.

In our computational hardness proofs, we will often include subsets of
alternatives in the descriptions of preference orders. For example, if
$A$ is the set of alternatives and $B$ is some nonempty strict subset
of $A$, then by saying that some agent has preference order of the
form $B \succ A-B$, we mean that this agent ranks all the alternatives
in $B$ ahead of all the alternatives outside of $B$, and that the
order in which this agent ranks alternatives within $B$ and within
$A-B$ is irrelevant (and, thus, one can assume any easily computable
order).

A \emph{positional scoring function} (PSF) is a function $\alpha^m:
[m] \rightarrow \naturals$. A PSF $\alpha^m$ is an \emph{increasing
  positional scoring function} (IPSF) if for each $i,j \in [m]$, if $i
< j$ then $\alpha^m(i) < \alpha(j)$. Analogously, a PSF $\alpha^m$ is
a \emph{decreasing positional scoring function} (DPSF) if for each
$i,j \in [m]$, if $i < j$ then $\alpha^m(i) > \alpha^m(j)$.

Intuitively, if $\beta^m$ is an IPSF then $\beta^m(i)$ gives the
\emph{dissatisfaction} that an agent suffers from when assigned to an
alternative that is ranked $i$'th on his or her preference
order. Thus, we assume that for each IPSF $\beta^m$ it holds that
$\beta^m(1) = 0$ (an agent is not dissatisfied by his or her top
alternative).  Similarly, a DPSF $\gamma^m$ measures an agent's
satisfaction and we assume that for each DPSF $\gamma^m$ it holds that
$\gamma^m(m)=0$. 

We will often speak of families $\alpha$ of IPSFs (DPSFs) of the form
$\{ \alpha^m \mid m \in \naturals, \alpha^m \mbox{ is a PSF}\}$, where
the following holds: 
\begin{enumerate}
\fixlist
\item For IPSFs, for each $m \in \naturals$ it
  holds that $(\forall i \in [m])[ \alpha^{m+1}(i) = \alpha^m(i)]$.
\item For DPSFs, for each $m \in \naturals$ it holds that $(\forall i
  \in [m])[ \alpha^{m+1}(i+1) = \alpha^m(i)]$.
\end{enumerate}
In other words, we build our families of IPSFs (DPSFs) by appending
(prepending) values to functions with smaller domains.  We assume that
each function $\alpha^m$ from a family can be computed in polynomial
time with respect to $m$.  To simplify notation, we will refer to such
families of IPSFs (DPSFs) as \emph{normal} IPSFs (normal
DPSFs).

We are particularly interested in normal IPSFs (normal DPSFs)
corresponding to the Borda count method. That is, in the families of
IPSFs $\alpha^{m}_{\bordainc}(i) = i-1$ (in the families of DPSFs
$\alpha^{m}_{\bordadec}(i) = m - i$).\medskip

\noindent
\textbf{Our Resource Allocation Problem.}\quad
We consider a problem of finding function $\Phi: N \rightarrow A$ that
assigns each agent to some alternative (we will call $\Phi$ an
\emph{assignment function}). We say that $\Phi$ is feasible if for
each alternative $a$ it holds that the total weight of the agents
assigned to it does not exceed its capacity $\capacity_a$.  Further,
we define the cost of assignment $\Phi$ to be $\cost(\Phi) = \sum_{a:
  \Phi^{-1}(a) \neq \emptyset}c_{a}$.

Given an IPSF (DPSF) $\alpha^m$, we consider two \emph{dissatisfaction
  functions}, $\ell_{1}^{\alpha}(\Phi)$ and
$\ell_\infty^\alpha(\Phi)$, (two \emph{satisfaction functions},
$\ell_{1}^{\alpha}(\Phi)$ and $\min^\alpha(\Phi)$), measuring the
quality of the assignment as follows:
\begin{enumerate}
\fixlist
\item $\ell_{1}^{\alpha}(\Phi) = \sum_{i=1}^{n}\alpha(\pos_{i}(\Phi(i)))$.
\item $\ell_{\infty}^{\alpha}(\Phi) = \mathrm{max}_{i =
    1}^{n}\alpha(\pos_{i}(\Phi(i)))$ (or, $\min^{\alpha}(\Phi) =
  \mathrm{min}_{i = 1}^{n}\alpha(\pos_{i}(\Phi(i)))$).
\end{enumerate}
The former one measures agents' total dissatisfaction (satisfaction),
whereas the latter one considers the most dissatisfied (the least
satisfied) agent only. In welfare economics and multiagent resource
allocation theory the two metrics correspond to, respectively,
utilitarian and egalitarian social welfare.
We define our resource allocation problem as follows.

\begin{definition}\label{def:assignment}
  Let $\alpha$ be a normal IPSF. An instance of
  $\alpha$-\textsc{Assignment-Inc} problem consists of a set of agents
  $N = [n]$, a set of alternatives $A = \{a_1, \ldots
  a_m\}$, a preference profile $V$ of the agents, a sequence $(w_1,
  \ldots, w_n)$ of agents' weights, sequences $(\capacity_{a_1},
  \ldots, \capacity_{a_m})$ and $(c_{a_1}, \ldots, c_{a_m})$ of
  alternatives' capacities and costs, respectively, and budget $B \in
  \naturals$. We ask for the assignment function $\Phi$ such
  that:
  \begin{enumerate}
  \fixlist
  \item $\cost(\Phi) \leq B$,
  \item $\forall_{a \in A} \sum_{i: \Phi(i) = a}w_{i} \leq cap_{a}$, and
  \item $\ell_{1}^\alpha(\Phi)$ is minimized.
  \end{enumerate}
\end{definition}

In other words, in $\alpha$-\textsc{Assignment-Inc} we ask for a
feasible assignment that minimizes the total dissatisfaction of the
agents without exceeding the budget.

Problem $\alpha$-\textsc{Assignment-Dec} is defined identically except
that $\alpha$ is a normal DPSF and in the third condition we seek to
maximize $\ell_{1}^\alpha(\Phi)$ (that is, in
$\alpha$-\textsc{Assignment-Dec} our goal is to maximize total
satisfaction). If we replace $\ell^\alpha_1$ with $\ell_\infty^\alpha$
in $\alpha$-\textsc{Assignment-Inc} then we obtain problem
$\alpha$-\textsc{Minmax-Assignment-Inc}, where we seek to minimize the
dissatisfaction of the most dissatisfied agent.  If we replace
$\ell^\alpha_1$ with $\min^\alpha$ in $\alpha$-\textsc{Assignment-Dec}
then we obtain problem $\alpha$-\textsc{Minmax-Assign\-ment-Dec},
where we seek to maximize the satisfaction of the least satisfied
agent.

As far as optimal solutions go, satisfaction and dissatisfaction
formulations of our problems are equivalent. However, as we will see,
there are striking differences in terms of their approximability.

Clearly, each of our four \textsc{Assignment} problems is
$\np$-complete: Even without costs they reduce to the standard
$\np$-complete \textsc{Partition} problem, where we ask if a set of
integers (in our case these integers would be agents' weights) can be
split evenly between two sets (in our case, two alternatives with the
capacities equal to half of the total agent weight). However, in very
many applications (for example, in the sport classes example from the
introduction) it suffices to consider unit-weight agents.  Thus, from
now on we assume the agents have unit weights.

Our four problems can be viewed as generalizations of
Monroe's~\cite{monroeElection} and Chamberlin and
Courant's~\cite{ccElection} multiwinner voting systems (see the
introduction for their definitions). To model Monroe's system, it
suffices to set the budget $B = K$, the cost of each alternative to be
$1$, and the capacity of each alternative to be $ \frac{\|N\|}{K}$
(for simplicity, throughout the paper we assume that $K$ divides
$\|N\|$). We will refer to thus restricted variants of our problems as
\textsc{Monroe-Assignment} variants.  To represent Chamberlin and
Courant's system within our framework, it suffices to take the same
restrictions as for Monroe's system, except that each alternative has
capacity equal to $\|N\|$.  We will refer to thus restricted variants
of our problems as \textsc{CC-Assignment} variants.

The eight above-defined special cases of our resource allocation
problem were, in various forms and shapes, considered by Procaccia,
Rosenschein, and Zohar~\cite{complexityProportionalRepr}, Lu and
Boutilier~\cite{budgetSocialChoice}, and Betzler, Slinko and
Uhlmann~\cite{fullyProportionalRepr}.\medskip

\noindent\textbf{Computational Complexity, Approximation Algorithms}.\quad
For many normal IPSFs $\alpha$ (and, in particular, for Borda count),
even the above-mentioned restricted versions of the original problem,
namely, $\alpha$-\textsc{Monroe-Assignment-Inc},
$\alpha$-\textsc{Minmax-Monroe-Assignment-Inc},
$\alpha$-\textsc{CC-Assignment-Inc}, and
$\alpha$-\textsc{Minmax-CC-Assignment-Inc} are
$\np$-complete\cite{fullyProportionalRepr, complexityProportionalRepr}
(the same holds for normal DPSFs and \textsc{Dec} variants of the
problems). Thus, we explore possibilities for approximate solutions.

\begin{definition}
  Let $\beta$ be a real number such that $\beta \geq 1$ ($0 < \beta
  \leq 1$) and let $\alpha$ be a normal IPSF (a normal DPSF).  An
  algorithm is a $\beta$-approximation algorithm for
  $\alpha$-\textsc{Assignment-Inc} problem (for
  $\alpha$-\textsc{Assignment-Dec} problem) if on each
  instance $I$ it returns a feasible assignment $\Phi$ that meets the
  budget restriction and such that $\ell_{1}^{\alpha}(\Phi) \leq \beta
  \cdot \OPT$ (and such that $\ell_{1}^{\alpha}(\Phi) \geq
  \beta\cdot \OPT$), where $\OPT$ is the aggregated dissatisfaction
  (satisfaction) $\ell_{1}^{\alpha}(\Phi_\OPT)$ of the optimal
  assignment $\Phi_\OPT$.
\end{definition}

We define $\beta$-approximation algorithms for the \textsc{Minmax}
variants of our problems analogously. For example, Lu and
Boutilier~\cite{budgetSocialChoice} present a $(1 -
\frac{1}{e})$-approximation algorithm for the case of
\textsc{CC-Assignment-Dec}.

Throughout this paper, we will consider
each of the \textsc{Monroe-Assignment} and \textsc{CC-Assignment}
variants of the problem and for each we will either prove
inapproximability with respect to any constant $\beta$ (under
standard complexity-theoretic assumptions) or we will show an approximation algorithm.  In our
inapproximability proofs, we will use the following classic
$\np$-complete problems~\cite{gar-joh:b:int}.

\begin{definition}
  An instance $I$ of \textsc{Set-Cover} consists of set $U = [n]$
  (called the ground set), family $\calF = \{F_{1}, F_{2}, \dots,
  F_{m}\}$ of subsets of $U$, and positive integer $K$. We ask if
  there exists a set $I \subseteq [m]$ such that $\|I\|
  \leq K$ and $\bigcup_{i\in I}F_i = U$.
\end{definition}

\begin{definition}
  \textsc{Vertex-Cover} is a special case of \textsc{Set-Cover}, where
  $U$ and $\calF$ are constructed from a given graph
  $G$. Specifically, $U$ is the set of $G$'s edges and $\calF = \{F_1,
  \ldots, F_n\}$ corresponds to $G$'s vertices (for each vertex $v$ of
  $G$, $\calF$ has a corresponding set $F$, which contains the edges
  incident to $v$).
\end{definition}

\begin{definition}
  \textsc{X3C} is a special case of \textsc{Set-Cover} where $\|U\|$
  is divisible by $3$, each member of $\calF$ has exactly three
  elements, and $K = \frac{n}{3}$.
\end{definition}


Note that \textsc{X3C} remains $\np$-complete even if we additionally
assume that $n$ is divisible by $2$ and each member of $U$ appears in
at most $3$ sets from $\calF$~\cite{gar-joh:b:int}.

\section{Hardness of Approximation}\label{sec:approximation}

In this section we present our inapproximability results for
\textsc{Monroe-Assignment} and \textsc{CC-Assignment} variants of the
resource allocation problem. In particular, we show that if we focus
on voter dissatisfaction (i.e., on the \textsc{Inc} variants) then for
each $\beta > 1$, neither Monroe's nor Chamberlin and Courant's system
has a polynomial-time $\beta$-approximation algorithm.  Further, we
show that analogous results hold if we focus on the satisfaction of
the least satisfied voter.

Naturally, these inapproximability results carry over to more general settings.
In particular, unless $\p = \np$, there are no polynomial-time constant-factor
approximation algorithms for the general resource allocation problem for the
case where we focus on voter dissatisfaction. On the other hand, our results
do not preclude good approximation algorithms for the case where we measure
agents' total satisfaction. Indeed, in Section~\ref{sec:algorithms} we derive algorithms
for several satisfaction-based special cases of the problem.

We start by showing that for each normal IPSF $\alpha$ there is no
constant-factor polynomial-time approximation algorithm for
$\alpha$-\textsc{Monroe-Assignment-Inc} (and, thus, there is no such
algorithm for general $\alpha$-\textsc{Assignment-Inc}).

\begin{theorem}\label{theorem:noApprox1}
  For each normal IPSF $\alpha$ and each constant factor $\beta$,
  $\beta > 1$, there is no polynomial-time $\beta$-approximation
  algorithm for $\alpha$-\textsc{Monroe-Assignment-Inc} unless $\p = \np$.
\end{theorem}
\begin{proof}
  Let us fix a normal IPSF $\alpha$ and let us assume, for the sake of
  contradiction, that there is some constant $\beta$, $\beta > 1$, and
  a polynomial-time $\beta$-approximation algorithm $\calA$ for
  $\alpha$-\textsc{Monroe-Assignment-Inc}.

  Let $I$ be an instance of \textsc{X3C} with ground set $U = [n]$ and
  family $\calF = \{F_{1}, F_{2}, \dots, F_{m}\}$ of $3$-element
  subsets of $U$. W.l.o.g., we assume that $n$ is divisible by both
  $2$ and $3$ and that each member of $U$ appears in at most 3 sets
  from $\calF$.

  Given $I$, we build instance $I_{M}$ of
  $\alpha$-\textsc{Monroe-Assignment-Inc} as follows.  We set $N = U$
  (that is, the elements of the ground set are the agents) and we set $A
  = A_1 \cup A_2$, where $A_1 = \{a_1, \ldots, a_m\}$ is a set of
  alternatives corresponding to the sets from the family $\calF$ and
  $A_2$, $\|A_2\| = \frac{n^{2} \cdot \alpha(3) \cdot \beta}{2}$, is a
  set of dummy alternatives needed for our construction. We let $m' =
  \|A_2\|$ and we rename the alternatives in $A_2$ so that $A_2 =
  \{b_1, \ldots, b_{m'}\}$.  We set $K = \frac{n}{3}$.

  We build agents' preference orders using the following algorithm.
  For each $j \in N$, set $M_f(j) = \{ a_i \mid j \in F_i\}$ and $M_l
  = \{ a_i \mid j \not\in F_i \}$. Set $m_f(j) = \|M_f(j)\|$ and
  $m_l(j) = \|M_l(j)\|$; as the frequency of the elements from $U$ is
  bounded by 3, $m_f(j) \leq 3$. For each agent $j$ we set his or her
  preference order to be of the form $M_f(j) \succ_j A_2 \succ_j
  M_l(j)$, where the alternatives in $M_f(j)$ and $M_l(j)$ are ranked
  in an arbitrary way and the alternatives from $A_2$ are placed at
  positions $m_{f}(j) + 1, \dots, m_{f}(j) + m'$ in the way described
  below (see Figure \ref{fig:diag1} for a high-level illustration of
  the construction).
  
  \begin{figure}[tb]
    \begin{center}
      \includegraphics[scale=0.75]{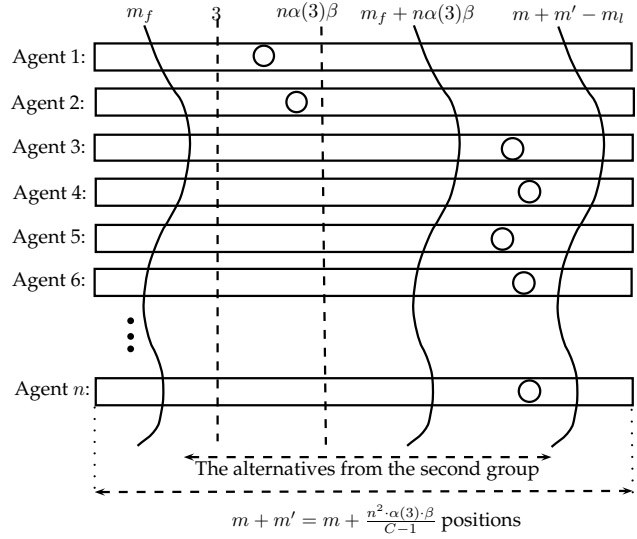}
    \end{center}
    \vspace{-0.75cm}
    \caption{The alignment of the positions in the preference orders
      of the agents. The positions are numbered from the left to the
      right. The left wavy line shows the positions $m_{f}(\cdot)$,
      each no greater than $3$. The right wavy line shows the
      positions $m_{l}(\cdot)$, each higher than $n \cdot \alpha(3)
      \cdot \beta$. The alternatives from $A_{2}$ (positions of one
      such an alternative is illustrated with the circle) are placed
      only between the peripheral wavy lines. Each alternative from
      $A_{2}$ is placed on the left from the middle wavy line exactly
      2 times, thus each such alternative is placed on the left from
      the right dashed line no more than $2$ times (exactly two times
      at the figure).}
    \label{fig:diag1}
  \end{figure}

  We place the alternatives from $A_{2}$ in the preference orders of
  the agents in such a way that for each alternative $b_i \in A_{2}$
  there are at most two agents that rank $b_i$ among their $n \cdot
  \alpha(3) \cdot \beta$ top alternatives.  The following construction
  achieves this effect.  If $(i + j) \,\mathrm{mod}\, n < 2$, then
  alternative $b_{i}$ is placed at one of the positions $m_{f}(j) + 1,
  \dots, m_{f}(j) + n \cdot \alpha(3) \cdot \beta$ in $j$'s preference
  order. Otherwise, $b_i$ is placed at a position with index higher
  than $m_{f}(j) + n \cdot \alpha(3) \cdot \beta$ (and, thus, at a
  position higher than $n \cdot \alpha(3)\cdot \beta$). This
  construction can be implemented because for each agent $j$ there are
  exactly $m' \cdot \frac{2}{n} = n \cdot \alpha(3) \cdot \beta$
  alternatives $b_{i_{1}}, b_{i_{2}}, b_{i_{n \alpha(3) \beta}}$ such
  that $(i_{k} + j) \,\mathrm{mod}\, n < 2$.

  Let $\Phi$ be an assignment computed by $\calA$ on $I_{M}$.  We
  will show that $\ell_{1}^{\alpha}(\Phi) \leq n \cdot \alpha(3) \cdot
  \beta$ if and only if $I$ is a \emph{yes}-instance of \textsc{X3C}.

  ($\Leftarrow$) If there exists a solution for $I$ (i.e., an exact
  cover of $U$ with $\frac{n}{3}$ sets from $\calF$), then we can
  easily show an assignment in which each agent $j$ is assigned to an
  alternative from the top $m_{f}(j)$ positions of his or her
  preference order (namely, one that assigns each agent $j$ to the
  alternative $a_i \in A_1$ that corresponds to the set $F_i$, from
  the exact cover of $U$, that contains $j$). Thus, for the optimal
  assignment $\Phi_\OPT$ it holds that $\ell_{1}^{\alpha}(\Phi_\OPT) \leq
  \alpha(3) \cdot n$. In consequence, $\mathcal{A}$ must return an
  assignment with the total dissatisfaction at most $n \cdot \alpha(3)
  \cdot \beta$.

  ($\Rightarrow$) Let us now consider the opposite direction. We
  assume that $\calA$ found an assignment $\Phi$ such that
  $\ell_{1}^{\alpha}(\Phi) \leq n \cdot \alpha(3) \cdot \beta$ and we
  will show that $I$ is a \emph{yes}-instance of \textsc{X3C}.  Since
  we require each alternative to be assigned to either $0$ or $3$
  agents, if some alternative $b_{i}$ from $A_2$ were assigned to some
  $3$ agents, at least one of them would rank him or her at a position
  worse than $n \cdot \alpha(3) \cdot \beta$. This would mean that
  $\ell_{1}^{\alpha}(\Phi) \geq n \cdot \alpha(3) \cdot \beta +
  1$. Analogously, no agent $j$ can be assigned to an alternative that
  is placed at one of the $m_l(j)$ bottom positions of $j$'s
  preference order. Thus, only the alternatives in $A_1$ have agents
  assigned to them and, further, if agents $x$, $y$, $z$, are assigned
  to some $a_i \in A_1$, then it holds that $F_i = \{x,y,z\}$ (we will
  call each set $F_i$ for which alternative $a_i$ is assigned to some
  agents $x,y,z$ \emph{selected}). Since each agent is assigned to
  exactly one alternative, the selected sets are disjoint. Since the
  number of selected sets is $K = \frac{n}{3}$, it must be the case
  that the selected sets form an exact cover of $U$. Thus, $I$ is a
  \emph{yes}-instance of \textsc{X3C}.
\end{proof}

One may wonder if hardness of approximation for
$\alpha$-\textsc{Monroe-Assignment-Inc} is not an artifact of the
strict requirements regarding the budget.  It turns out that unless
$\p = \np$, there is no $\beta$-$\gamma$-approximation algorithm that
finds an assignment with the following properties: (1) the aggregated
dissatisfaction $\ell_{1}^{\alpha}(\Phi)$ is at most $\beta$ times
higher than the optimal one, (2) the number of alternatives to which
agents are assigned is at most $\gamma K$ and (3) each selected alternative (the
alternative that has agents assigned), is assigned to no more than $\gamma \lceil
\frac{n}{K} \rceil$ and no less than $\frac{1}{\gamma} \lceil
\frac{n}{K}\rceil$ agents. (The proof is similar to the one used for
Theorem~\ref{theorem:noApprox1}.)  Thus, in our further study we do
not consider relaxations of the budget constraints.


\begin{theorem}\label{theorem:noApprox2}
  For each normal IPSF $\alpha$ and each constant $\beta$, $\beta >
  1$, there is no polynomial-time $\beta$-approximation algorithm for
  $\alpha$-\textsc{Minmax-Monroe-Assignment-Inc} unless $\p = \np$.
\end{theorem}
\begin{proof}
  The proof of Theorem~\ref{theorem:noApprox1} applies to this case as
  well. (In fact, it even suffices to take $m' = \|A_2\| = \frac{n
    \cdot \alpha(3) \cdot \beta}{2}$).
\end{proof}

Results analogous to Theorems~\ref{theorem:noApprox1}
and~\ref{theorem:noApprox2} hold for the \textsc{CC-Assginment-Inc}
family of problems as well. 

\begin{theorem}\label{theorem:noApprox3}
  For each normal IPSF $\alpha$ and each constant factor $\beta$,
  $\beta > 1$, there is no polynomial-time $\beta$-approximation
  algorithm for $\alpha$-\textsc{CC-Assignment-Inc} unless $\p = \np$.
\end{theorem}

\begin{proof}
  Let us fix a normal IPSF $\alpha$.  For the sake of contradiction,
  let us assume that there is some constant $\beta$, $\beta > 1$, and a
  $\beta$-approximation algorithm $\calA$ for
  $\alpha$-\textsc{CC-Assignment-Inc}. We will show that it is
  possible to use $\calA$ to solve the $\np$-complete
  \textsc{Vertex-Cover} problem.

  Let $I = (U, \calF, K)$ be an instance of \textsc{Vertex-Cover},
  where $U = [n]$ is the ground set, $\calF = \{F_1, \ldots, F_m\}$ is
  a family of subsets of $U$ (where each member of $U$ belongs to exactly
  two sets in $\calF$), and $K$ is a positive integer. 

  Given $I$, we construct an instance $I_{CC}$ of
  $\alpha$-\textsc{CC-Assignment-Inc} in the following way: The set of
  agents is $N = U$ and the set of alternatives is $A =
  \bigcup_{j=1}^m A_j$, where each $A_j$ contains exactly $\alpha(2)
  \cdot \beta \cdot n$ (unique) alternatives. Intuitively, for each
  $j$, $1 \leq j \leq m$, the alternatives in $A_j$ correspond to the set
  $F_j$. For each $A_j$, $1 \leq j \leq m$, we pick one alternative,
  which we denote $a_j$.  For each agent $i \in N$, we set $i$'s
  preference order as follows: Let $F_j$ and $F_k$, $j < k$, be the
  two sets that contain $i$. Agent $i$'s preference order is of the
  form $a_j \succ_i a_k \succ_i A_k - \{a_k\} \succ_i A - (A_k \cup
  \{a_j,a_k\})$ (the particular order of alternatives in the sets
  $A_k-\{a_k\}$ and $A - (A_k \cup \{a_j,a_k\})$ is irrelevant for the
  construction). We ask for an assignment of the agents to at most $K$
  alternatives.

  %

  Let us consider a solution $\Phi$ returned by $\mathcal{A}$ on input
  $I_{CC}$. We claim that $\ell_{1}^{\alpha}(\Phi) \leq n \cdot
  \alpha(2) \cdot \beta$ if and only if $I$ is a \emph{yes}-instance
  of \textsc{Vertex-Cover}.

  ($\Leftarrow$) If $I$ is a \emph{yes}-instance then, clearly, each
  agent $i$ can be assigned to one of the top two alternatives in his
  or her preference order (if there is a size-$K$ cover, then this
  assignment selects at most $K$ candidates).  Thus the total
  dissatisfaction of an optimal assignment is at most $\alpha(2) \cdot
  n$. As a result, the solution $\Phi$ returned by $\calA$ has total
  dissatisfaction at most $\alpha(2) \cdot \beta \cdot n$.

  ($\Rightarrow$) If $\mathcal{A}$ returns an assignment with total
  dissatisfaction no greater than $\alpha(2) \cdot \beta \cdot n$,
  then, by the construction of agents preference orders, we see that
  each agent $i$ was assigned to an alternative from a set $A_j$ such
  that $i \in F_j$. Since the assignment can use at most $K$ alternatives,
  this directly implies that there is a size-$K$ cover of $U$ with sets
  from $\calF$.
\end{proof}

\begin{theorem}\label{theorem:noApprox4}
  For each normal IPSF $\alpha$ and each constant factor $\beta$,
  $\beta > 1$, there is no polynomial-time $\beta$-approximation
  algorithm for $\alpha$-\textsc{Minmax-CC-Assignment-Inc} unless $\p
  = \np$.
\end{theorem}
\begin{proof}
  The proof of Theorem~\ref{theorem:noApprox3} works correctly in this
  case as well. In fact, it even suffices to take the $m$ groups of
  alternatives, $A_1, \ldots, A_m$, to contain $\alpha(2) \cdot \beta$
  alternatives each.
\end{proof}

The above results show that approximating the minimal dissatisfaction
of agents is difficult. On the other hand, if we focus on agents'
total satisfaction then constant-factor approximation exist in many
cases (see, e.g., the work of Lu and
Boutilier~\cite{budgetSocialChoice} and the next section). Yet, if we
focus on the satisfaction of the least satisfied voter, there are no
efficient constant-factor approximation algorithms for Monroe's and
Courant and Chamberlin's systems. (However, note that our result for
the Monroe setting is more general than the result for the
Chamberlin-Courant setting; the latter is for the Borda DPSF only.)

\begin{theorem}\label{theorem:noApprox5}
  For each normal DPSF $\alpha$ (where each entry is polynomially
  bounded in the number of alternatives) and each constant factor
  $\beta$, $0 < \beta \leq 1$, there is no $\beta$-approximation
  algorithm for $\alpha$-\textsc{Minmax-Monroe-Assignment-Dec} unless
  $\p = \np$.
\end{theorem}
\begin{proof}
  Let us fix a DPSF $\alpha$ where each entry is polynomially bounded
  in the number of alternatives. For the sake of contradiction, let us
  assume that for some $\beta$, $0 < \beta \leq 1$, there is a
  $\beta$-approximation algorithm $\mathcal{A}$ for
  $\alpha$-\textsc{Minmax-Monroe-Assignment-Dec}. We will show that
  the existence of this algorithm implies that \textsc{X3C} is
  solvable in polynomial time.

  Let $I$ be an \textsc{X3C} instance with ground set $U = \{1, 2,
  \dots, n\}$ and a collection $\calF = \{F_{1}, \dots, F_{m}\}$ of
  subsets of $U$. Each set in $\calF$ has cardinality three. Further,
  w.l.o.g., we can assume that $n$ is divisible by three and that each
  $i \in U$ appears in at most three sets from $\calF$.  Given $I$, we
  form an instance $I_M$ of
  $\alpha$-\textsc{Minmax-Monroe-Assignment-Dec} as follows.  Let $n'
  = 3 \cdot (\alpha_{dec}^{m+1}(1) \cdot \lceil \frac{1 - \beta}{\beta}\rceil +
  3)$.  The set $N$ of agents is partitioned into two subsets, $N_1$
  and $N_2$. $N_1$ contains $n$ agents (intuitively, corresponding to
  the elements of the ground set $U$) and $N_2$ contains $n'$ agents
  (used to enforce certain properties of the solution).  The set of
  alternatives $A$ is partitioned into two subsets, $A_1$ and $A_2$.
  We set $A_1 = \{a_1, \ldots, a_m\}$ (members of $A_1$ correspond to
  the sets in $\calF$), and we set $A_2 = \{b_1, \ldots, b_{m'}\}$,
  where $m' = \frac{n'}{3}$.

  For each $j$, $1 \leq j \leq n$, we set $M_f(j) = \{ a_i \mid j \in
  F_i\}$.  For each $j$, $1 \leq j \leq n$, we set the preference
  order of the $j$'th agent in $N_1$ to be of the form 
  \[ 
     M_f(j) \succ A_2 \succ A_1 - M_f(j).\] 
  Note that by our assumptions, $\|M_f(j)\| \leq 3$.
  For each $j$, $1 \leq j \leq n'$, we set the preference order of the
  $j$'th agent in $N_2$ to be of the form
  \[ 
     b_{\left\lceil\frac{j}{3}\right\rceil} \succ A_2 -
     \{b_{\left\lceil\frac{j}{3}\right\rceil}\} \succ A_1.
  \]
  Note that each agent in $N_2$ ranks the alternatives from $A_1$ in
  positions positions $m'+1, \ldots, m'+m$.  Finally, we set the
  budget/number of candidates that can be selected, to be $K =
  \frac{n+n'}{3}$.

  Now, consider the solution $\Phi$ returned by $\mathcal{A}$ on
  $I_{M}$. We will show that $\ell_{\infty}^{\alpha_{dec}^{m +
      m'}}(\Phi) \leq$ $\beta\alpha_{dec}^{m + m'}(3)$ if and only
  if $I$ is a \emph{yes}-instance of \textsc{X3C}.

  ($\Leftarrow$) If there exists an exact set cover of $U$ with sets
  from $\calF$, then it is easy to construct a solution for $I_M$
  where the satisfaction of each agent is greater or equal to
  $\beta\cdot\alpha_{dec}^{m + m'}(3)$. Let $I \subseteq \{1, \ldots, m\}$
  be a set such that $\bigcup_{i \in I}F_i = U$ and $\|I\| =
  \frac{n}{3}$. We assign each agent $j$ from $N_1$ to the alternative
  $a_i$ such that (a) $i \in I$ and (b) $j \in F_i$, and we assign
  each agent from $N_2$ to his or her most preferred alternative.
  Thus, Algorithm $\calA$ has to return an assignment with the minimal
  satisfaction greater or equal to $\beta\cdot\alpha_{dec}^{m + m'}(3)$.

  ($\Rightarrow$) For the other direction, we first show that
  $\beta\cdot\alpha_{dec}^{m + m'}(3) \geq \alpha_{dec}^{m + m'}(m')$.
  Since DPSFs are strictly decreasing, it holds that:
  \begin{equation}
    \label{eq:1}
    \beta\cdot\alpha_{dec}^{m + m'}(3) \geq \beta\cdot(\alpha_{dec}^{m + m'}(m') + m' - 3).
  \end{equation}
  Then, by the definition of DPSFs, it holds that:
  \begin{equation}
    \label{eq:2}
    \alpha_{dec}^{m + m'}(m') = \alpha_{dec}^{m + 1}(1)
  \end{equation}
  Using the fact that $m' = (\alpha_{dec}^{m+1}(1) \cdot \lceil
  \frac{1 - \beta}{\beta}\rceil + 3)$ and using~\eqref{eq:2},
  we can transform inequality~\eqref{eq:1} to obtain the following:
  \begin{align*}
    \beta\cdot\alpha_{dec}^{m + m'}(3) &\geq \beta\cdot(\alpha_{dec}^{m + m'}(m') + m' - 3) \\
       & = \beta\cdot\left(\alpha_{dec}^{m + m'}(m') + (\alpha_{dec}^{m+1}(1) \cdot \left\lceil\frac{1 - \beta}{\beta}\right\rceil + 3) - 3\right)\\
       & \geq \beta\cdot\alpha_{dec}^{m + m'}(m') + (1-\beta)\cdot \alpha_{dec}^{m+1}(1) \\
       &  = \beta\cdot\alpha_{dec}^{m + m'}(m') + (1-\beta)\cdot \alpha_{dec}^{m + m'}(m') = \alpha_{dec}^{m + m'}(m').
  \end{align*}

  This means that if the minimal satisfaction of an agent is at least
  $\beta\cdot\alpha_{dec}^{m + m'}(3)$, then no agent was assigned to
  an alternative that he or she ranked beyond position $m'$.  If some
  agent $j$ from $N_{1}$ were assigned to an alternative from $A_{2}$,
  then, by the pigeonhole principle, some agent from $N_{2}$ were
  assigned to an alternative from $A_{1}$. However, each agent in
  $N_2$ ranks the alternatives from $A_1$ beyond position $m'$ and
  thus such an assignment is impossible.  In consequence, it must be
  that each agent in $j$ was assigned to an alternative that
  corresponds to a set $F_{i}$ in $\calF$ that contains $j$. Such an
  assignment directly leads to a solution for $I$.
\end{proof}

Let us now move on to the case of \textsc{Minmax-CC-Assignment-Dec}
family of problems. Unfortunately, in this case our inapproximability
argument holds for the case of Borda DPSF only (though we believe that
it can be adapted to other DPSFs as well). Further, in our previous
theorems we were showing that existence of a respective
constant-factor approximation algorithm implies that $\np$ collapses
to $\p$. In the following theorem we will show a seemingly weaker
collapse of $\wtwo$ to $\fpt$.

Intuitively, $\fpt$ is a class of problems that can be solved in time
$f(k)n^{O(1)}$, where $n$ is the size of the input instance, $k$ is a
so-called parameter (some quantity, typically characterizing the
difficulty of the instance), and $f$ is some computable function.  For
example, for the \textsc{Set-Cover} and \textsc{Vertex-Cover} problems
one often takes $K$ as the value of the parameter.  In the world of
parametrized complexity, $\fpt$ is viewed as the class of \emph{easy}
problems (analogous to the class $\p$), whereas classes $\wone
\subseteq \wtwo \subseteq \cdots$ are believed to form a hierarchy of
classes of \emph{hard} problems (somewhat analogous to the class
$\np$).  It holds that $\fpt \subseteq \wone$, but it seems unlikely
that $\fpt = \wone$, let alone $\fpt = \wtwo$.  We point the reader to
the books of Niedermeier~\cite{nie:b:invitation-fpt} and Flum and
Grohe~\cite{flu-gro:b:parameterized-complexity} for detailed overviews
of parametrized complexity theory.  Interestingly, while both
\textsc{Set-Cover} and \textsc{Vertex-Cover} are $\np$-complete, the
former is $\wtwo$-complete and the latter belongs to $\fpt$ (see,
e.g., the book of Niedermeier~\cite{nie:b:invitation-fpt} for these
now-standard results and their history).

To prove hardness of approximation for
$\alpha_\bordadec$-\textsc{Minmax-CC-Assignment-Dec}, we first prove
the following simple lemma.

\begin{lemma}\label{lemma:coveringSubsets}
  Let $K, p, l$ be three positive integers and let $X$ be a set of
  cardinality $lpK$. There exists a family $\calS = \{S_1, \ldots,
  S_{\binom{lK}{K}} \}$ of $pK$-element subsets of $X$ such that for
  each $K$-element subset $B$ of $X$, there is a set $S_i \in \calS$
  such that $B \subseteq S_i$.
\end{lemma}
\begin{proof}
  Set $X' = [lK]$ and let $Y'$ be a family of all $K$-element subsets
  of $X'$. Replace each element $i$ of $X'$ with $p$ new elements (at
  the same time replacing $i$ with the same $p$ elements within each
  set in $Y'$ that contains $i$). As a result we obtain two new sets,
  $X$ and $Y$, that satisfy the statement of the theorem (up to the
  renaming of the elements).
\end{proof}

\begin{theorem}\label{theorem:noApprox6}
  Let $\alpha^{m}_{\bordadec}$ be the Borda DPSF
  ($\alpha^{m}_{\bordadec}(i) = m - i$). For each constant factor
  $\beta$, $0 < \beta \leq 1$, there is no $\beta$-approximation
  algorithm for
  $\alpha^{m}_{\bordadec}$-\textsc{Minmax-CC-Assignment-Dec} unless
  $\fpt = \wtwo$.
\end{theorem}
\begin{proof}
  For the sake of contradiction, let us assume that there is some
  constant $\beta$, $0 < \beta \leq 1$, and a polynomial-time
  $\beta$-approximation algorithm $\calA$ for
  $\alpha^{m}_{\bordadec}$-\textsc{Minmax-CC-Assignment-Dec}. We will
  show that the existence of this algorithm implies that
  \textsc{Set-Cover} is fixed-parameter tractable for the parameter
  $K$ (since \textsc{Set-Cover} is known to be $\wtwo$-complete for
  this parameter, this will imply $\fpt=\wtwo$).

  Let $I$ be an instance of \textsc{Set-Cover} with ground set $U =
  [n]$ and family $\calF = \{F_{1}, F_{2}, \dots, F_{m}\}$ of subsets
  of $U$.  Given $I$, we build an instance $I_{CC}$ of
  $\alpha^{m}_{\bordadec}$-\textsc{Minmax-CC-Assignment-Dec} as
  follows. The set of agents $N$ consists of $n$ subsets of agents,
  $N_1, \ldots, N_n$, where each group $N_i$ contains exactly $n' =
  \binom{\left\lceil \frac{2}{\beta} \right\rceil K}{K}$ agents.
  Intuitively, for each $i$, $1 \leq i \leq n$, the agents in the set
  $N_{i}$ correspond to the element $i$ in $U$.  The set of
  alternatives $A$ is partitioned into two subsets, $A_1$ and $A_2$,
  such that $A_1 = \{a_1, \ldots, a_m\}$ is a set of alternatives
  corresponding to the sets from the family $\calF$, and $A_2$,
  $\|A_2\| = \left\lceil \frac{2}{\beta}\right\rceil \left\lceil
    \frac{m(1 + \beta)}{K} \right\rceil K$, is a set of dummy
  alternatives needed for our construction.  We set $m' = \|A\| = m +
  \|A_2\|$.

  Before we describe the preference orders of the agents in $N$, we
  form a family $R = \{r_1, \ldots, r_{n'}\}$ of preference orders
  over $A_2$ that satisfies the following condition: For each
  $K$-element subset $B$ of $A_2$, there exists $r_j$ in $R$ such that
  all members of $B$ are ranked among the bottom $\left\lceil \frac{m(1 +
    \beta)}{K} \right\rceil K$ positions in $r_j$. By
  Lemma~\ref{lemma:coveringSubsets}, such a construction is possible
  (it suffices to take $l = \left\lceil \frac{2}{\beta}\right\rceil$ and $p =
  \left\lceil \frac{m(1 + \beta)}{K} \right\rceil$); further, the proof of the
  lemma provides an algorithmic way to construct $R$.

  We form the preference orders of the agents as follows.  For each
  $i$, $1 \leq i \leq n$, set $M_f(i) = \{ a_t \mid i \in F_t\}$. For
  each $i$, $1 \leq i \leq n$, and each $j$, $1 \leq j \leq n'$, the
  $j$'th agent from $N_i$ has preference order of the form:
  \[
     M_f(i) \succ r_j \succ A_1 - M_f(i)
  \]
  (we pick any arbitrary, polynomial-time computable order of
  candidates within $M_f(i)$ and $M_l(i)$).

  Let $\Phi$ be an assignment computed by $\calA$ on $I_{M}$.  We will
  show that $\ell_{\infty}^{\alpha^{m'}_{\bordadec}}(\Phi) \geq
  \beta\cdot(m' - m)$ if and only if $I$ is a \emph{yes}-instance of
  \textsc{Set-Cover}.

  ($\Leftarrow$) If there exists a solution for $I$ (i.e., a cover of
  $U$ with $K$ sets from $\calF$), then we can easily show an
  assignment where each agent is assigned to an alternative that he or
  she ranks among the top $m$ positions (namely, for each $j$, $1 \leq
  j \leq n$, we assign all the agents from the set $N_j$ to the
  alternative $a_i \in A_1$ such that $j \in F_i$ and $F_i$ belongs to
  the alleged $K$-element cover of $U$). Under this assignment, the
  least satisfied agent's satisfaction is at least $m'-m$ and, thus,
  $\calA$ has to return an assignment $\Phi$ where
  $\ell_{\infty}^{\alpha^{m'}_{\bordadec}}(\Phi) \geq \beta\cdot(m' -
  m)$.

  ($\Rightarrow$) Let us now consider the opposite direction. We
  assume that $\calA$ found an assignment $\Phi$ such that
  $\ell_{\infty}^{\alpha^{m}_{\bordadec}}(\Phi) \geq \beta\cdot(m' -
  m)$ and we will show that $I$ is a \emph{yes}-instance of
  \textsc{Set-Cover}.  We claim that for each $i$, $1 \leq i \leq n $,
  at least one agent $j$ in $N_i$ were assigned to an alternative from
  $A_{1}$. If all the agents in $N_i$ were assigned to alternatives
  from $A_2$, then, by the construction of $R$, at least one of them
  would have been assigned to an alternative that he or she ranks at a
  position greater than $\|A_2\| - \left\lceil \frac{m(1 + \beta)}{K}\right\rceil
  K = \left\lceil \frac{2}{\beta}\right\rceil \left\lceil \frac{m(1 + \beta)}{K}
  \right\rceil K - \left\lceil \frac{m(1 + \beta)}{K}\right\rceil K$.  Since for $x =
  \left\lceil \frac{m(1 + \beta)}{K} \right\rceil K$ we have:
  \begin{align*}
    \left\lceil \frac{2}{\beta}\right\rceil x - x \geq m' - m'\beta +
    m\beta
  \end{align*}
  (we skip the straightforward calculation).
  This means that this agent would have been assigned to an
  alternative that he or she ranks at a position greater than $m' -
  m'\beta + m\beta$.  As a consequence, this agent's satisfaction
  would be lower than $(m' - m)\beta$.
  Similarly, no agent from $N_{i}$ can be assigned to an alternative
  from $M_l(i)$. Thus, for each $i$, $1 \leq i \leq n$, there exists
  at least one agent $j \in N_{i}$ that is assigned to an alternative
  from $M_f(i)$.  In consequence, the covering subfamily of $\calF$
  consists simply of those sets $F_k$, for which some agent is
  assigned to alternative $a_k \in A_1$.

  The presented construction gives the exact algorithm for
  \textsc{Set-Cover} problem running in time $f(K)(n+m)^{O(1)}$, where
  $f(K)$ is polynomial in $\binom{\left\lceil \frac{2}{\beta}
    \right\rceil}{K}$. The existence of such an algorithm means that
  \textsc{Set-Cover} is in $\fpt$. On the other hand, we know that
  \textsc{Set-Cover} is $\wtwo$-complete, and thus if $\calA$ existed
  then $\fpt = \wtwo$ would hold.
\end{proof}


\section{Approximation Algorithms}
\label{sec:algorithms}

We now turn to approximation algorithms for Monroe's and Chamberlin
and Courant's multiwinner voting rules (both of which are special
cases of our resource allocation problem). Indeed, if one focuses on
agents' total satisfaction then it is possible to obtain high-quality
approximation results.  In particular, we show the first nontrivial
(randomized) approximation algorithm for
$\alpha_{\bordadec}$-\textsc{Monroe-Assignment-Dec} (for each
$\epsilon > 0$, we can provide a randomized polynomial-time algorithm
that achieves $0.715 - \epsilon$ approximation ratio), and the first
polynomial-time approximation scheme (PTAS) for
$\alpha_{\bordadec}$-\textsc{CC-Assignment-Dec}. These results stand
in a sharp contrast to those from the previous section, where we have
shown that approximation is hard for essentially all remaining
variants of the problem.

The core difficulty in solving $\alpha$-\textsc{Monroe/CC-Assignment}
problems lays in selecting the alternatives that should be assigned to
the agents. Given a preference profile and a set $A'$ of $K$
alternatives, using a standard network-flow argument, it is easy to
match them optimally to the agents.
\begin{proposition}[\textbf{Implicit in the paper of Betzler et
    al.~\cite{fullyProportionalRepr}}]\label{prop:assignment}
  Let $\alpha$ be a normal DPSF, $N$ be a set of agents, $A$ be a set
  of alternatives, $V$ be a preference profile of $N$ over $A$, and
  $A'$ a $K$-element subset of $A$ (where $K$ divides $\|N\|$). There
  is a polynomial-time algorithm that computes an optimal assignment
  $\Phi$ of the alternatives from $A'$ to the agents, both for
  $\ell_1^\alpha$ and for $\min^\alpha$, both for the case where each
  alternative in $A'$ should be assigned to the same number of agents
  (Monroe's case) and for the case without additional restrictions
  (Chamberlin and Courant's case).
\end{proposition}
Thus, in the algorithms in this section, we will focus on the issue of
selecting the alternatives and not on the issue of matching them to
the agents.

\subsection{Monroe's System}\label{sec:monroe-approx}

We first consider $\alpha_\bordadec$-\textsc{Monroe-Assignment-Dec}.
%
%
%
%
%
%
%
%
Perhaps the most natural approach to solve this problem is to build a
solution iteratively: In each step we pick some not-yet-assigned
alternative $a_i$ (using some criterion) and assign him or her to
those $\lceil \frac{N}{K} \rceil$ agents that (a) are not assigned to
any other alternative yet, and (b) whose satisfaction of being matched
with $a_i$ is maximal. It turns out that this idea, implemented
formally in Algorithm~\ref{alg:greedy}, works very well in many
cases. We provide a lower bound on the total satisfaction it
guarantees in the next lemma.  (For each positive integer $k$, we let
$H_k = \sum_{i=1}^k\frac{1}{i}$ be the $k$'th harmonic number. Recall
that $H_k = \Theta(\log k)$.)

\SetKwInput{KwNotation}{Notation}
\begin{algorithm}[t]
   \small
   \SetAlCapFnt{\small}
   \KwNotation{$\Phi \leftarrow$ a map defining a partial assignment, iteratively built by the algorithm. \\
          $\hspace{21pt}$ $\Phi^{\leftarrow} \leftarrow$ the set of agents for which the assignment is 
          already defined. \\
          $\hspace{21pt}$ $\Phi^{\rightarrow} \leftarrow$ the set of alternatives already used in the 
          assignment.}
   \If{$K \leq 2$}{compute the optimal solution using an algorithm of Betzler et al.~\cite{fullyProportionalRepr} and return.} 
   $\Phi = \{\}$ \\
   \For{$i\leftarrow 1$ \KwTo $K$}{
      $score \leftarrow \{\}$ \\
      $bests \leftarrow \{\}$ \\
      \ForEach{$a_{i} \in A \setminus \Phi^{\rightarrow}$}{
          $agents \leftarrow$ sort $N \setminus \Phi^{\leftarrow}$ so that $j \prec k$ in $agents$ $\implies$ \\
                           $\hspace{37pt}$ $pos_{j}(a_{i}) \leq pos_{k}(a_{i})$ \\
          $bests[a_{i}] \leftarrow$ chose first $\lceil \frac{N}{K} \rceil$ elements from $agents$ \\
          $score[a_{i}] \leftarrow \sum_{j \in bests} (m - pos_{j}(a_{i}))$\\
      }
      $a_{best} \leftarrow \mathrm{argmax}_{a \in A \setminus \Phi^{\rightarrow}} score[a]$ \\
      \ForEach{$j \in bests[a_{best}]$}{
         $\Phi[j] \leftarrow a_{best}$ \\
      }
   }
   \caption{\small The algorithm for \textsc{Monroe-Assignment}.}
   \label{alg:greedy}
\end{algorithm}

\begin{lemma}\label{lemma:greedy}
  Algorithm~\ref{alg:greedy} is a polynomial-time $(1 -
  \frac{K-1}{2(m-1)} - \frac{H_K}{K})$-approximation algorithm for
  $\alpha_{\bordadec}$-\textsc{Monroe-Assignment-Dec}.
\end{lemma}
\begin{proof}
  Our algorithm computes an optimal solution for $K \leq 2$. Thus we
  assume $K \geq 3$.  Let us consider the situation in the algorithm
  after the $i$'th iteration of the outer loop (we have $i=0$ if no
  iteration has been executed yet). So far, the algorithm has picked
  $i$ alternatives and assigned them to $i\frac{n}{K}$ agents (recall
  that for simplicity we assume that $K$ divides $n$ evenly). Hence,
  each agent has $\lceil \frac{m-i}{K-i} \rceil$ unassigned
  alternatives among his or her $i+ \lceil \frac{m-i}{K-i} \rceil$
  top-ranked alternatives.  By pigeonhole principle, this means that
  there is an unassigned alternative $a_{\ell}$ who is ranked among
  top $i+ \lceil \frac{m-i}{K-i} \rceil$ positions by at least
  $\frac{n}{K}$ agents. To see this, note that there are
  $(n-i\frac{n}{K})\lceil \frac{m-i}{K-i} \rceil$ slots for unassigned
  alternatives among the top $i+ \lceil \frac{m-i}{K-i} \rceil$
  positions in the preference orders of unassigned agents, and that
  there are $m-i$ unassigned alternatives. As a result, there must be
  an alternative $a_\ell$ for whom the number of agents that rank him
  or her among the top $i+ \lceil \frac{m-i}{K-i} \rceil$ positions is
  at least:
  \[
       \frac{1}{m-i}\left((n-i\frac{n}{K})\lceil \frac{m-i}{K-i} \rceil\right) \geq
       \frac{n}{m-i}\left(\frac{K-i}{K}\right)\left(\frac{m-i}{K-i}\right)  
       =\frac{n}{K}.
  \]
%
%
%
  In consequence, the $\lceil \frac{n}{K} \rceil$ agents assigned in
  the next step of the algorithm will have the total satisfaction at
  least $\lceil \frac{n}{K} \rceil \cdot (m - i - \lceil
  \frac{m-i}{K-i} \rceil)$. Thus, summing over the $K$ iterations, the
  total satisfaction guaranteed by the assignment $\Phi$ computed by
  Algorithm~\ref{alg:greedy} is at least the following value (see the
  comment below for the fourth inequality; for the last inequality we
  assume $K \geq 3$):
  \begin{align*}
    \ell_{1}^{\alpha_{b}}(\Phi) & \geq \sum_{i = 0}^{K-1} \frac{n}{K}  \cdot \left(m - i - \lceil \frac{m-i}{K-i} \rceil\right) \\
    & \geq \sum_{i = 0}^{K-1} \frac{n}{K}  \cdot \left( m - i - \frac{m-i}{K-i} -1 \right) \\
    & = \sum_{i = 1}^{K} \frac{n}{K} \cdot \left(m - i - \frac{m-1}{K-i+1} + \frac{i-2}{K-i+1} \right) \\
    & = \frac{n}{K}\left( \frac{K(2m-K-1)}{2} -(m-1) H_K + K(H_K-1) - H_K \right) \\
    & = (m-1)n \left( 1 - \frac{K-1}{2(m-1)} - \frac{H_K}{K} + \frac{H_K-1}{m-1} - \frac{H_K}{K(m-1)}  \right)\\
    & > (m-1)n \left( 1 - \frac{K-1}{2(m-1)} - \frac{H_K}{K} \right)
%
  \end{align*}
  The fourth equality holds because:
  \begin{align*}
     K(H_K-1)-H_K 
      &=  \sum_{i=1}^{K}\left(\frac{K}{i}-1\right) - H_K 
      =  \sum_{i=1}^{K}\left(\frac{K}{K-i+1}-1\right) - H_K \\
      &= \sum_{i=1}^{K}\frac{i-1}{K-i+1} - H_K 
      = \sum_{i=1}^{K}\frac{i-2}{K-i+1}.
  \end{align*}
  If each agent were assigned to his or her top alternative, the total
  satisfaction would be equal to $(m-1)n$. Thus we get the following
  bound:
\begin{align*}
  \frac{\ell_{1}^{\alpha_{\bordadec}}(\Phi)}{\OPT} \leq 1 - \frac{K-1}{2(m-1)} - \frac{H_K}{K}. 
\end{align*}
This completes the proof.
\end{proof}

Note that in the above proof we measure the quality of our assignment
against a perhaps-impossible, perfect solution, where each agent is
assigned to his or her top alternative.  This means that for
relatively large $m$ and $K$, and small $\frac{K}{m}$ ratio, the
algorithm can achieve a close-to-ideal solution irrespective of the
voters' preference orders. We believe that this is an argument in
favor of using Monroe's system in multiwinner elections.  On the flip
side, to obtain a better approximation ratio, we would have to use a
more involved bound on the quality of the optimal solution. To see
that this is the case, form an instance $I$ of
$\alpha_\bordadec$-\textsc{Monroe-Assignment-Dec} with $n$ agents and
$m$ alternatives, where all the agents have the same preference order,
and where the budget is $K$ (and where $K$ divides $n$).  It is easy
to see that each solution that assigns the $K$ universally top-ranked
alternatives to the agents is optimal.  Thus the total dissatisfaction
of the agents in the optimal solution is:
\begin{align*}
  \frac{n}{K}\left( (m-1) + \cdots + (m-K) \right) &= \frac{n}{K}
  \left(\frac{K(2m-K-1)}{2}\right) 
   = n(m-1) \left( 1 - \frac{K-1}{2(m-1)} \right).
\end{align*}
By taking large enough $m$ and $K$ (even for a fixed value of
$\frac{m}{K}$), the fraction $1 - \frac{K-1}{2(m-1)}$ can be
arbitrarily close to the approximation ratio of our algorithm (the
reasoning here is somewhat in spirit of the idea of identifying
maximally robust elections, as studied by Shiryaev, Yu, and
Elkind~\cite{shi-yu-elk:t:robust-winners}).

Betzler et al.~\cite{fullyProportionalRepr} showed that for each fixed
constant $K$, $\alpha_\bordadec$-\textsc{Monroe-Assignment-Dec} can be
solved in polynomial time. Thus, for small values of $K$ for which the
fraction $\frac{H_K}{K}$ affects the approximation guarantees of
Algorithm~\ref{alg:greedy} too much, we can use this polynomial-time
algorithm to find an optimal solution.  This means that we can
essentially disregard the $\frac{H_K}{K}$ part of
Algorithm~\ref{alg:greedy}'s approximation ratio.  In consequence, the
quality of the solution produced by Algorithm~\ref{alg:greedy} most
strongly depends on the ratio $\frac{K-1}{m-1}$. In most cases we can
expect it to be small (for example, in Polish parliamentary elections
$K = 460$ and $m \approx 6000$; in this case the greedy algorithm's
approximation ratio is about $0.96$).  For the remaining cases, for
example, when $K > \frac{m}{2}$, we can use a simple sampling-based
randomized algorithm described below.

The idea of this algorithm is to randomly pick $K$ alternatives and
match them optimally to the agents, using
Proposition~\ref{prop:assignment}. Naturally, such an algorithm might
be very unlucky and pick $K$ alternatives that all of the agents rank
low. Yet, if $K$ is large relative to $m$ then it is likely that such
a random sample would include a large chunk of some optimal solution.
In the lemma below, we asses the expected satisfaction obtained with a
single sampling step (relative to the satisfaction given by the
optimal solution) and the probability that a single sampling step
gives satisfaction close to the expected one.  Naturally, in practice
one should try several sampling steps and pick the one with the
highest satisfaction.


\begin{lemma}
  A single sampling step of the randomized algorithm for
  $\alpha_\bordadec$-\textsc{Monroe-Assignment-Dec} achieves expected
  approximation ratio of $\frac{1}{2}(1 + \frac{K}{m} - \frac{K^2}{m^2-m} +
  \frac{K^3}{m^3-m^2})$. Let $p_{\epsilon}$ denote the
  probability that the relative deviation between the obtained total
  satisfaction and the expected total satisfaction is higher than
  $\epsilon$; for $K \geq 8$ we have $p_{\epsilon} \leq \exp \left(-
    \frac{K\epsilon^2}{128} \right)$.
\end{lemma}
\begin{proof}
  Let $N = [n]$ be the set of agents, $A = \{a_1, \ldots, a_m\}$ be
  the set of alternatives, and $V$ be the preference profile of the
  agents.  Let us fix some optimal solution $\Phi_\opt$ and let
  $A_\opt$ be the set of alternatives assigned to the agents in this
  solution.  For each $a_{i} \in A_{\opt}$, we write $\sat(a_{i})$ to
  denote the total satisfaction of the agents assigned to $a_{i}$ in
  $\Phi_\opt$. Naturally, we have $\sum_{a \in A_{\opt}} \sat(a) =
  \OPT$.  In a single sampling step, we choose uniformly at random a
  $K$-element subset $B$ of $A$.  Then, we form a solution $\Phi_B$ by
  matching the alternatives in $B$ optimally to the agents (via
  Proposition~\ref{prop:assignment}).  We write $K_{\opt}$ to denote
  the random variable equal to $\|A_\opt \cap B\|$, the number of
  sampled alternatives that belong to $A_{\opt}$. We define $p_{i} =
  \Pr(K_{\opt} = i)$. For each $j \in \{1, \ldots, K\}$, we write
  $X_j$ to denote the random variable equal to the total satisfaction
  of the agents assigned to the $j$'th alternative from the sample.
  We claim that for each $i$, $0 \leq i \leq K$, it holds that:
  \begin{align*}
    \E\left(\sum_{j=1}^{K}X_{j} | K_\opt = i\right) \geq \frac{i}{K}\OPT + \frac{m-i-1}{2}
    \cdot (n - i\frac{n}{K}).
  \end{align*}
  Why is this so? Given sample $B$ that contains $i$ members of
  $A_\opt$, our algorithm's solution is at least as good as a solution
  that matches the alternatives from $B \cap A_\opt$ in the same way
  as $\Phi_\opt$, and the alternatives from $B - A_\opt$ in a random
  manner.  Since $K_\opt = i$ and each $a_j \in A_\opt$ has equal
  probability of being in the sample, it is easy to see that the
  expected value of $\sum_{a_j \in B \cap A_\opt}\sat(a_j)$ is $
  \frac{i}{K}\OPT$.
  %
  %
  After we allocate the agents from $B \cap A_\opt$, each of the
  remaining, unassigned agents has $m-i$ positions in his or her
  preference order where he ranks the agents from $A - A_\opt$.  For
  each unassigned agents, the average score value associated with
  these positions is at least $\frac{m-i-1}{2}$ (this is so, because
  in the worst case the agent could rank the alternatives from $B \cap
  A_\opt$ in the top $i$ positions). There are $(n - i\frac{n}{K})$
  such not yet assigned agents and so the expected total satisfaction
  from assigning them randomly to the alternatives is $\frac{m-i-1}{2}
  \cdot (n - i\frac{n}{K})$. This proves our bound on the expected
  satisfaction of a solution yielded by optimally matching a random
  sample of $K$ alternatives.


  Since $\OPT$ is upper bounded by $(m-1)n$ (consider a
  possibly-nonexistent solution where every agent is assigned to his
  or her top preference), we get that:
  \[
    \E\left(\sum_{j=1}^{K}X_{j} | K_\opt = i\right) \geq  \frac{i}{K}\OPT + \frac{m-i-1}{2(m-1)} \cdot (1 - \frac{i}{K})\OPT.
  \]
  We can compute the unconditional expected satisfaction of $\Phi_B$
  as follows:
  \begin{align*}
    \E\left(\sum_{j=1}^{K}X_{j}\right) & = \sum_{i=0}^{K}p_{i}\E\left(\sum_{j=1}^{K}X_{j} | K_\opt = i\right) \\
    & \geq \sum_{i=0}^{K}p_{i}\left(\frac{i}{K}\OPT + \frac{m-i-1}{2(m-1)}
    \cdot (1 - \frac{i}{K})\OPT\right).
  \end{align*}
  Since $\sum_{i=1}^{K}p_{i} \cdot i$ is the expected number of the
  alternatives in $A_{\opt}$, we have that $ \sum_{i=1}^{K}p_{i} \cdot
  i = \frac{K^2}{m}.  $ (one can think of summing the expected values
  of $K$ indicator random variables; one for each element of $A_\opt$,
  taking the value $1$ if a given alternative is selected and taking
  the value $0$ otherwise).  Further, from the generalized mean
  inequality we obtain $\sum_{i=1}^{K}p_{i} \cdot i^{2} \geq
  \left(\frac{K^2}{m}\right)^{2}.$ In consequence, through routine
  calculation, we get that:
  \begin{align*}
     \E\left(\sum_{j=1}^{K}X_{j}\right) 
     & \geq \left(\frac{K}{m}\OPT + \frac{m^2 - K^2 -m}{2m(m-1)} \cdot \left(1 - \frac{K}{m}\right)\OPT\right) \\
     & = \frac{\OPT}{2}\left(1 + \frac{K}{m} - \frac{K^2}{m^2-m} +
     \frac{K^3}{m^3-m^2}\right).
  \end{align*}

  It remains to assess the probability that the total satisfaction
  obtained through $\Phi_B$ is close to its expected value.  Since
  $X_{j} \in \langle 0, \frac{(m-1)n}{K} \rangle$, from Hoeffding's
  inequality we get:
  \begin{align*}
    p_{\epsilon} & = \Pr\left(\left|\sum_{j=1}^{K}X_{j} - \E(\sum_{j=1}^{K}X_{j})\right| \geq \epsilon \E(\sum_{j=1}^{K}X_{j})\right) \\
    & \leq \exp \left(- \frac{2\epsilon^2 (\E(\sum_{j=1}^{K}X_{j}))^{2}}{K(\frac{(m-1)n}{K})^2}
    \right) = \exp \left(- \frac{K\epsilon^2 (\E(\sum_{j=1}^{K}X_{j}))^{2}}{((m-1)n)^2} \right)
  \end{align*}
  We note that since $\frac{K}{m}-\frac{K^2}{m^2-m} \geq 0$, our
  previous calculations show that $\E(\sum_{j=1}^{K}X_{j}) \geq
  \frac{\OPT}{2}$. Further, for $K \geq 8$, Lemma~\ref{lemma:greedy}
  (and the fact that in its proof we upper-bound $\OPT$ to be $(m-1)n$)
  gives that $\OPT \geq \frac{mn}{8}$. Thus $p_{\epsilon} \leq \exp
  \left(- \frac{K\epsilon^2}{128} \right)$.  This completes the proof.
\end{proof}

The threshold for $\frac{K}{m}$, where the randomized algorithm is (in
expectation) better than the greedy algorithm is about 0.57.  Thus, by
combining the two algorithms, we can guarantee an expected
approximation ratio of $0.715 - \epsilon$, for each fixed constant
$\epsilon$.  The pseudo-code of the combination of the two algorithms
is presented in Algorithm~\ref{alg:combination}.

\begin{theorem}
  For each fixed $\epsilon$, Algorithm~\ref{alg:combination} provides
  a $(0.715-\epsilon)$-approximate solution for the problem
  $\alpha_\bordadec$-\textsc{Monroe-Assignment-Dec} with probability
  $\lambda$ in time polynomial with respect to the input instance
  size and $-\log(1-\lambda)$.
\end{theorem}
\begin{proof}
  Let $\epsilon$ be a fixed constant. We are given an instance $I$ of
  $\alpha_\bordadec$-\textsc{Monroe-Assignment-Dec}. If $m \leq 1 +
  \frac{2}{\epsilon}$, we solve $I$ using a brute-force algorithm
  (note that in this case the number of alternatives is at most a
  fixed constant). Similarly, if $\frac{H_K}{K} \geq
  \frac{\epsilon}{2}$ then we use the exact algorithm of Betzler et
  al.~\cite{fullyProportionalRepr} for a fixed value of $K$ (note that
  in this case $K$ is no greater than a certain fixed constant). We do
  the same if $K \leq 8$.

  On the other hand, if neither of the above conditions hold, we try
  both Algorithm~\ref{alg:greedy} and a number of runs of the
  sampling-based algorithm. It is easy to check through routine
  calculation that if $\frac{H_K}{K} \leq \frac{\epsilon}{2}$ and $m >
  1 + \frac{2}{\epsilon}$ then Algorithm~\ref{alg:greedy} achieves
  approximation ratio no worse than $(1 -\frac{K}{2m} - \epsilon)$.
  %
  %
  %
  %
  We run the sampling-based algorithm $\frac{-512 \log (1 -
    \lambda)}{K\epsilon^2}$ times. The probability that a single run
  fails to find a solution with approximation ratio at least
  $\frac{1}{2}(1 + \frac{K}{m} - \frac{K^2}{m^2-m} +
  \frac{K^3}{m^3-m^2}) - \frac{\epsilon}{2}$ is
  $p_{\frac{\epsilon}{2}} \leq \exp \left(- \frac{K\epsilon^2}{4 \cdot
      128} \right)$. Thus, the probability that at least one run will
  find a solution with at least this approximation ratio is at least:
  \begin{align*}
    1 - p_{\frac{\epsilon}{2}}^{\frac{-512 \log (1 - \lambda)}{K\epsilon^2}} = 1
    - \exp \left(-\frac{K\epsilon^2}{4\cdot128} \cdot \frac{-512 \log
        (1 - \lambda)}{K\epsilon^2} \right) = \lambda.
  \end{align*}
  Since $m \leq 1 + \frac{2}{\epsilon}$, by routine calculation we see
  that the sampling-based algorithm with probability $\lambda$ finds a
  solution with approximation ratio at least $\frac{1}{2}(1 +
  \frac{K}{m} - \frac{K^2}{m^2} + \frac{K^3}{m^3}) - \epsilon$.  By
  solving the equality:
  \begin{align*}
    \frac{1}{2}\left(1 + \frac{K}{m} - \frac{K^2}{m^2} + \frac{K^3}{m^3}\right) =
    1 -\frac{K}{2m}
  \end{align*}
  we can find the value of $\frac{K}{m}$ for which the two algorithms
  give the same approximation ratio.  By substituting $x =
  \frac{K}{m}$ we get equality $1 + x - x^2 + x^3 = 2 - x$.  One can
  calculate that this equality has a single solution within $\langle
  0,1 \rangle$ and that this solution is $x \approx 0.57$. For this
  $x$ both algorithms guarantee approximation ratio of $0.715 -
  \epsilon$. For $x < 0.57$ the deterministic algorithm guarantees a
  better approximation ratio and for $x > 0.57$, the randomized
  algorithm does better.
\end{proof}

\SetKwInput{KwNotation}{Notation}
\begin{algorithm}[t]
   \small
   \SetAlCapFnt{\small}
   \SetKwInOut{Parameters}{Parameters}
   \KwNotation{We use the same notation as in Algorithm~\ref{alg:greedy}; $\w(\cdot)$ denotes the Lambert W-Function.\\}
   \Parameters{$\lambda$ $\leftarrow$ required probability of achieving the approximation ratio equal $0.715 - \epsilon$}
   \If{$\frac{H_K}{K} \geq
  \frac{\epsilon}{2}$ 
  or $K \leq 8$}{compute the optimal solution using an algorithm of Betzler et al.~\cite{fullyProportionalRepr} and return.}
   \If{$m \leq 1 + \frac{2}{\epsilon}$}{compute the optimal solution using a simple brute force algorithm and return.} 
   $\Phi_1 \leftarrow$ solution returned by Algorithm~\ref{alg:greedy} \\
   $\Phi_2 \leftarrow$ run the sampling-based algorithm $\frac{-512 \log (1 - \lambda)}{K\epsilon^2}$ times; select the assignment of the best quality \\
   return the better assignment among $\Phi_1$ and $\Phi_2$
   \caption{\small The combination of two algorithms for \textsc{Monroe-Assignment}.}
   \label{alg:combination}
\end{algorithm}

\subsection{Chamberlin and Courant's System}\label{sec:CCAlgorithm}

Let us now move on to the Chamberlin and Courant's system.  Since
Chamberlin and Courant's system places fewer restrictions on the
solution assignment than the Monroe's system, our algorithms for
$\alpha_\bordadec$-\textsc{Monroe-Assignement-Dec} can also be used
for $\alpha_\bordadec$-\textsc{CC-Assignement-Dec}, improving upon the
approximation algorithms of Lu and
Boutilier~\cite{budgetSocialChoice}. However, it turns out that the
additional freedom of Chamberlin and Courant's system allows us to go
even further, and to design a polynomial-time approximation scheme for
$\alpha_\bordadec$-\textsc{CC-Assignement-Dec}.



The idea of our algorithm (presented as Algorithm~\ref{alg:greedy2}
below) is to compute a certain value $x$ and to greedily compute an
assignment that (approximately) maximizes the number of agents
assigned to one of their top-$x$ alternatives. If after this process
some agent has no alternative assigned, we assign him or her to his or
her most preferred alternative from those already picked.  Somewhat
surprisingly, it turns out that this greedy strategy achieves
high-quality results. (Recall that for nonnegative
real numbers, Lambert's W function, $\w(x)$, is defined to be the
solution of the equality $x = \w(x)e^{\w(x)}$.)


\begin{algorithm}[t]
   \small
   \SetAlCapFnt{\small}
   \KwNotation{We use the same notation as in Algorithm~\ref{alg:combination};
   $\hspace{21pt$} $\mathrm{num\_pos}_x(a) \leftarrow \|\{i \in [n] \setminus \Phi^{\leftarrow} : pos_i(a) \leq x \}\|$ (the number of not-yet assigned agents that rank alternative $a$ in one of their first $x$ positions)}
   $\Phi = \{\}$ \\
   $x = \lceil \frac{m\w(K)}{K} \rceil$ \\
   \For{$i\leftarrow 1$ \KwTo $K$}{
                        $a_{i} \leftarrow \mathrm{argmax}_{a \in A \setminus \Phi^{\rightarrow}} \mathrm{num\_pos}_x(a)$

      \ForEach{$j \in [n] \setminus \Phi^{\leftarrow}$}{
         \If{$pos_j(a_{i}) < x$}{
             $\Phi[j] \leftarrow a_{i}$ \\
         }
      }
   }
   \ForEach{$j \in A \setminus \Phi^{\leftarrow}$}{
       $a \leftarrow$ such server from $\Phi^{\rightarrow}$ that $\forall_{a' \in \Phi^{\rightarrow}} pos_{j}(a) \leq pos_{j}(a')$ \\
       $\Phi[j] \leftarrow a$ \\
   }
   \caption{\small The algorithm for \textsc{CC-Assignment}.}
   \label{alg:greedy2}
\end{algorithm}

\begin{lemma}
  Algorithm~\ref{alg:greedy2} is a polynomial-time $(1 -
  \frac{2\w(K)}{K})$-approximation algorithm for
  $\alpha_\bordadec$-\textsc{CC-Assignement-Dec}.
\end{lemma}
\begin{proof}
  Let $x = \frac{m\w(K)}{K}$.  We will first give an inductive proof
  that, for each $i$, $0 \leq i \leq K$, after the $i$'th iteration of the outer loop at
  most $n(1- \frac{w(K)}{K})^{i}$ agents are unassigned. Based on this
  observation, we will derive the approximation ratio of our
  algorithm.

  For $i = 0$, the inductive hypothesis holds because $n(1-
  \frac{\w(K)}{K})^{0} = n$. For each $i$, let $n_{i}$ denote the
  number of unassigned agents after the $i$'th iteration. Thus, after
  the $i$'th iteration there are $n_i$ unassigned agents, each with
  $x$ unassigned alternatives among his or her top-$x$ ranked
  alternatives. As a result, at least one unassigned alternative is
  present in at least $\frac{n_{i}x}{m-i}$ of top-$x$ positions of
  unassigned agents. This means that after the $(i+1)$'st iteration
  the number of unassigned agents is:
  \begin{align*}
    n_{i+1} \leq n_{i} - \frac{n_{i}x}{m-i} \leq n_{i}\left(1 -
    \frac{x}{m}\right) = n_{i}\left(1 - \frac{\w(K)}{K}\right).
  \end{align*}
  If for a given $i$ the inductive hypothesis holds, that is, if $n
  _{i} \leq n\left(1- \frac{\w(K)}{K}\right)^{i}$, then:
  \begin{align*}
    n_{i+1} \leq n(1- \frac{\w(K)}{K})^{i}(1 -
    \frac{\w(K)}{K}) = n(1- \frac{\w(K)}{K})^{i+1}
  \end{align*}
  Thus the hypothesis holds and, as a result, we have that:
  \begin{align*}
    n_{k} \leq n\left(1- \frac{\w(K)}{K}\right)^{K} \leq
    n\left(\frac{1}{e}\right)^{\w(K)} = \frac{n\w(K)}{K}.
  \end{align*}
  Let $\Phi$ be the assignment computed by our algorithm.  To compare
  it against the optimal solution, it suffices to observe that the
  optimal solution has satisfaction at most $\OPT \leq (m-1)n$,
  that each agent selected during the first $K$ steps has satisfaction
  at least $m-x = m-\frac{m\w(K)}{K}$, and that the agents not
  assigned within the first $K$ steps have satisfaction no worse than
  $0$. Thus it holds that:
  \begin{align*}
    \frac{\ell_{1}^{\alpha_{\bordadec}}(\Phi)}{\OPT} & \geq \frac{(n - \frac{n\w(K)}{K})(m - \frac{m\w(K)}{K})}{(m-1)n} \\
    & \geq (1 - \frac{\w(K)}{K})(1 - \frac{\w(K)}{K}) \geq 1 -
    \frac{2\w(K)}{K}
  \end{align*}
  This completes the proof.
\end{proof}

Since for each $\epsilon > 0$ there is a value $K_\epsilon$ such that
for each $K > K_\epsilon$ it holds that $\frac{2\w(K)}{K} < \epsilon$,
and $\alpha_\bordadec$-\textsc{CC-Assignment} problem can be solved
optimally in polynomial time for each fixed constant $K$(see the work
of Betzler et al.~\cite{fullyProportionalRepr}), there is a
polynomial-time approximation scheme (PTAS) for
$\alpha_\bordadec$-\textsc{CC-Assignment} (i.e., a family of
algorithms such that for each fixed $\beta$, $0 < \beta < 1$, there is
a polynomial-time $\beta$-approximation algorithm for
$\alpha_\bordadec$-\textsc{CC-Assignment} in the family; note that in
PTASes we measure the running time by considering $\beta$ to be a
fixed constant).

\begin{theorem}\label{theorem:ptas}
  There is a PTAS for $\alpha_\bordadec$-\textsc{CC-Assignment}.
\end{theorem}

The idea used in Algorithm~\ref{alg:greedy2} can also be used to
address a generalized \textsc{Minmax-CC-Assignment-Dec} problem. We
can consider the following relaxation of
\textsc{Minmax-CC-Assignment-Dec}: Instead of requiring that each
agent's satisfaction is lower-bounded by some value, we ask that the
satisfactions of a significant majority of the agents are
lower-bounded by a given value.  More formally, for a given constant
$\delta$, we introduce an additional quality metric:
\[ \min_{\delta}^{\alpha}(\Phi) = \mathrm{max}_{N' \subseteq N:
  \frac{||N|| - ||N'||}{||N||} \leq \delta}\mathrm{min}_{i \in
  N'}\alpha(pos_{i}(\Phi(i))).
\]
For a given $0 < \delta < 1$, by putting $x =
\frac{-m\ln(\delta)}{K}$, we get $(1 +
\frac{\ln(\delta)}{K})$-approximation algorithm for the
$\min_{\delta}^{\alpha}(\Phi)$ metric.

It would also be natural to try a sampling-based approach to solving
$\alpha_\bordadec$-\textsc{CC-Assignement-Dec}, just as we did for the
Monroe variant. Indeed, as recently (and independently) observed by
Oren~\cite{ore:p:cc}, this leads to a randomized algorithm with
expected approximation ratio of $(1 - \frac{1}{K+1})(1+\frac{1}{m})$.




\section{Conclusions}
\label{sec:conclusions}

\begin{table}[t!]
\centering
\footnotesize
\begin{tabular}{p{0.2cm}|p{1.6cm}|p{3.9cm}|p{1.6cm}|p{1.6cm}|p{1.6cm}|p{1.6cm}|}
\cline{2-7}
& \multicolumn{2}{|c|}{Monroe} & \multicolumn{2}{|c|}{Chamberlin and Courant}  &\multicolumn{2}{|c|}{General model}\\ \cline{2-7}
& dissat. & satisfaction & dissat. & satisfaction & dissat. & satisfaction \\ \cline{1-7}
\multicolumn{1}{|c|}{\begin{rotate}{270}\;\;\;\;\;\;utilitarian\end{rotate}}
& Inapprox. $\newline$ Theorem~\ref{theorem:noApprox1}
& Results for Borda only. $\newline$
randomized algorithm: $\newline$ 
(a) $0.715 - \epsilon$; 
$\newline$
(b) $\frac{1 + \frac{K}{m} - \frac{K^2}{m^2-m} + \frac{K^3}{m^3-m^2}}{2} -\epsilon$ 
deterministic algorithm: $\newline$
$1 - \frac{K-1}{2(m-1)}-\epsilon$
& Inapprox. $\newline$ Theorem~\ref{theorem:noApprox3}
& for Borda: $\newline$ PTAS $\newline$ (Sec.~\ref{sec:CCAlgorithm}); $\newline$ for general PSF: $\newline$ $(\frac{e-1}{e})$~\cite{budgetSocialChoice}
& Inapprox. $\newline$ Theorem~\ref{theorem:noApprox1} Theorem~\ref{theorem:noApprox3}
& Open problem \\ \cline{1-7}
\multicolumn{1}{|c|}{\begin{rotate}{270}egal.\end{rotate}}
& Inapprox. $\newline$ Theorem~\ref{theorem:noApprox2}
& Inapprox. $\newline$ Theorem~\ref{theorem:noApprox5}
& Inapprox. $\newline$ Theorem~\ref{theorem:noApprox4}
& Inapprox. $\newline$ Theorem~\ref{theorem:noApprox6}
& Inapprox. $\newline$ Theorem~\ref{theorem:noApprox2} Theorem~\ref{theorem:noApprox4}
& Inapprox. $\newline$ Theorem~\ref{theorem:noApprox5} Theorem~\ref{theorem:noApprox6}
\\ \cline{1-7}
\end{tabular}
\caption{Summary of approximation results fo Monroe's and 
  Chamberlin-Courant's multiwinner voting systems and for 
  our general resource allocation problem.}
\label{table:summary}
\end{table}

We have defined a certain resource allocation problem and have shown
that it generalizes multiwinner voting rules of Monroe and of
Chamberlin and Courant. Since it is known that these voting rules are
hard to
compute~\cite{complexityProportionalRepr,budgetSocialChoice,fullyProportionalRepr},
we focused on approximation properties of our problems. We have shown
that if we focus on agents' dissatisfaction then our problems are hard
to approximate up to any constant factor. The same holds for the case
where we focus on least satisfied agent's satisfaction. However, for
the case of optimizing total satisfaction, we have shown good
approximate solutions for the special cases of the resource allocation
problem pertaining to Monroe's and Chamberlin and Courant's voting
systems. In particular, we have shown a randomized algorithm that for
Monroe's system achieves approximation ratio arbitrarily close to
$0.715$, and for Chamberlin and Courant's system we have shown a
polynomial-time approximation scheme (PTAS). Table~\ref{table:summary}
summarizes the presented results.  The most pressing, natural open
question is whether there is a PTAS for Monroe's system. We plan to
address the problem of finding approximate solution in the general
framework for the case of optimizing total satisfaction and to
empirically verify the quality of our algorithms.

\medskip
\noindent\textbf{Acknowledgements} The
authors were supported in part by AGH University of Science and
Technology grant 11.11.120.865, by the Foundation for Polish Science's
Homing/Powroty program, and by EU's Human Capital Program "National
PhD Programme in Mathematical Sciences" carried out at the University
of Warsaw.

\bibliographystyle{plain}
\bibliography{main}

\end{document}